\documentclass{article}
\usepackage{ijcai22}

\usepackage{times}
\usepackage{soul}
\usepackage{url}
\usepackage[hidelinks]{hyperref}
\usepackage[utf8]{inputenc}
\usepackage[small]{caption}
\usepackage{graphicx}
\usepackage{amsmath}
\usepackage{amsthm}
\usepackage{booktabs}
\usepackage{algorithm}
\usepackage{algorithmic}
\usepackage{xspace}
\newtheorem{example}{Example}
\newtheorem{theorem}{Theorem}

\usepackage{xcolor}

\newcommand{\name}[0]{McSplit+LL\xspace}
\title{A Strengthened Branch and Bound Algorithm \\
for the Maximum Common (Connected) Subgraph Problem}

\author{
    Jianrong Zhou$^1$, %jrzhou@hust.edu.cn
    Kun He\textsuperscript{$\dag$}$^1$,  %brooklet60@hust.edu.cn
    Jiongzhi Zheng$^1$, %jzzheng@hust.edu.cn    
    Chu-Min Li$^2$,  %chu-min.li@u-picardie.fr
    Yanli Liu$^3$ %yanlil2008@163.com
    \affiliations
    $^1$School of Computer Science and Technology, Huazhong University of Science and Technology\\
    $^2$MIS, University of Picardie Jules Verne, France\\
    $^3$WuHan University of Science and Technology, China \\
    \emails
    \{jrzhou,jzzheng,brooklet60\}@hust.edu.cn \\
    \vspace{0.15em}
    \textsuperscript{$\dag$}Corresponding author.    
}

\begin{document}

\maketitle

\begin{abstract}
We propose a new and strengthened Branch-and-Bound (BnB) algorithm for the maximum common (connected) induced subgraph problem based on two new operators, Long-Short Memory (LSM) and Leaf vertex Union Match (LUM). Given two graphs for which we search for the maximum common (connected) induced subgraph,
the first operator of LSM maintains a score for the branching node using the short-term reward of each vertex of the first graph and the long-term reward of each vertex pair of the two graphs. In this way, the BnB process learns to reduce the search tree size significantly and improve the algorithm performance. 
The second operator of LUM further improves the performance by simultaneously matching the leaf vertices connected to the current matched vertices, and allows the algorithm to match multiple vertex pairs without affecting the solution optimality.
We incorporate the two operators into the state-of-the-art BnB algorithm McSplit, and denote the resulting algorithm as \name. Experiments show that \name outperforms 
McSplit+RL, a more recent variant of McSplit using reinforcement learning that is superior than McSplit.  
\end{abstract}

\section{Introduction}
%A graph is a typical model abstracted from many real-world applications, it consists of a set of vertices and a set of edges, the vertices indicate different objects and the edges indicate the relation between the vertices. Many common subgraph problems are proposed for measuring the similarity of two graphs. 

Graphs are important data structures that can be used to model many real-world problems. One of the most basic graph problems is to measure the similarity of graphs.
Given two graphs $G_0$ and $G_1$, the Maximum Common induced Subgraph (MCS) aims to find an induced subgraph $G_0'$ in $G_0$ and an induced subgraph $G_1'$ in $G_1$ such that $G_0'$ and $G_1'$ are isomorphic and the number of vertices of $G_0'$ and $G_1'$ is maximized. The vertices of $G_0'$ are said to be matched with the vertices of $G_1'$. MCS has a variant called the Maximum Common Connected induced Subgraph (MCCS), which further requires the induced subgraph to be connected. 
%A closely related problem called the induced Subgraph Isomorphism (SI) consists in determining whether a pattern graph appears in a target graph. 
These problems are widely applied in various domains, such as biochemistry~\cite{giugno2013,bonnici2013subgraph}, cheminformatics~\cite{raymond2002maximum,englert2015efficient,duesbury2017maximum}, compilers~\cite{blindell2015modeling}, model checking~\cite{sevegnani2015bigraphs}, molecular science~\cite{ehrlich2011maximum}, pattern recognition~\cite{solnon2015complexity}, malware detection~\cite{bruschi2006detecting,park2013deriving}, video and image analysis~\cite{shearer2001video,jiang2003image,hati2016image}, etc.

As an NP-hard problem, MCS is computationally very challenging. Many approaches have been proposed for solving MCS and its related problems, which could be mainly divided into two categories: exact algorithms and inexact algorithms. 
An exact algorithm guarantees to obtain an optimal solution but runs in exponential time in the worst case. It aims to efficiently enumerate the whole search space or reduce the search space without affecting the solution's optimality. The approaches for exact algorithms include Linear Programming (LP)~\cite{bahiense2012maximum}, Branch-and-Bound (BnB)~\cite{raymond2002maximum,MccreeshPT17,liu2020learning}, etc. 
An inexact algorithm aims to find a near-optimal solution within reasonable computational resources (e.g., time and memory). The approaches for inexact algorithms include heuristics~\cite{bonnici2013subgraph,englert2015efficient,duesbury2017maximum} and metaheuristics~\cite{choi2012efficient}.
Recently, machine learning~\cite{shearer2001video} and deep learning ~\cite{zanfir2018deep,bai2021glsearch} techniques are also adopted for solving MCS.

The BnB algorithms have exhibited high performance for MCS problems. Given a vertex $p$ in $G_0$, McSplit~\cite{MccreeshPT17}, one of the state-of-the-art algorithms for MCS, proposes a partition method to efficiently filter the set of candidate vertices of $G_1$ that can be matched with $q$, and uses a novel compact candidate set representation to dramatically reduce the memory and computational requirements during the search.
McSplit+RL~\cite{liu2020learning} further improves McSplit using a new branching method based on reinforcement learning.

%It rewards the branching vertices after matching a vertex pair and chooses the vertex with the maximum score as the branching node, then breaks ties using vertex degree on the backtracking. The learned rewards for branching find the best solution as early as possible to significantly reduce the search tree size. 
%is a learning based Reinfocement Learning (RL) instead of the vertex degree for BnB algorithm, and propose a new variant algorithm called McSplit+RL~\cite{liu2020learning}, which is the  state-of-the-art exact algorithm for MCS.McSplit+RL aims to reduce the size of the search tree and reach the pruning condition as early as possible. When the BnB algorithm matches a vertex pair, it receives the reduced value of the bound immediately and uses the reduced value to score each vertex or vertex pair. Then McSplit+RL uses the score of vertices in the branching heuristic where higher score vertex has higher priority for the vertex selection.  

In this work, we propose two new operators to speed up the BnB process, namely Long-Short Memory (LSM) and Leaf vertex Union Match (LUM). In a BnB process for MCS, the branching consists in first selecting a vertex $p$ from the first graph $G_0$, and then matching each candidate vertex $q$ of $G_1$ in turn.  
LSM maintains a score for each $p$ using the short-term reward, and the long-term reward of each branching pair $\langle p, q \rangle$. In this way, the search tree size could be reduced. 
LUM improves the efficiency in a different way.
When a pair of vertices are matched for branching, LUM also matches the leaf vertices connected to the current matched vertices, allowing to reduce the search space while keeping the solution's optimality.

We implement our two operators on top of McSplit, and design a new and strengthened algorithm called \name. 
We extensively evaluate \name on 25,552 instances from diverse applications as McSplit+RL does. 
Results show that the proposed \name clearly outperforms McSplit and McSplit+RL, which are already highly efficient. We also carry out an empirical analysis to give insight on why and how our two proposed operators are effective, suggesting that both of them can reduce the search tree size, so as to improve the efficiency.
Besides, the second operator of LUM is general and could be useful for other graph search problems.

\section{Problem Definition}%Preliminaries}
Consider a simple (undirected or directed , unlabelled)
%\HK{Add in future work that we can easily extend the method to handle labelled graph} 
graph $G = (V, E)$, where $V$ and $E$ represent the vertex set and the edge set, respectively. 
Two vertices $u,v$ are called adjacent if $(u,v) \in E$, and the degree of a vertex $u$ is the number of its adjacent vertices.
An induced subgraph $G' = (V', E')$ of $G$ consists of a vertex subset $V' \subseteq V$ and the edges in $E' = \{(u, v) \mid \forall u, v \in V', (u, v) \in E\}$. 

Given two graphs $G_0$ and $G_1$, 
%where $G_0$ is called the pattern graph and $G_1$ is called the target graph. 
if there is an induced subgraph $G_0'$ of $G_0$ and there exists a bijection $\phi : V_0'$ $\rightarrow V_1'$, $\Phi(G_0')$ ($\Phi(G)$ donates a graph that mapping all vertices in $G$ by $\phi$) is also an induced subgraph of $G_1$, then we call $\langle G_0', \Phi(G_0') \rangle$ is the common induced subgraph of $\langle G_0, G_1 \rangle$.
%If there are a pair of two graphs $\langle G_0, G_1 \rangle$, a subgraph $G'$, and it exists a vertex-to-vertex bijection $\phi : V_0 \rightarrow V_1$, the $G'$ is an induced graph of $G_0$ and $\Phi(G')$ ($\Phi$ is a bijection donated by mapping all the vertices in graph by $\phi$) also is an induced graph of $G_1$, we donate $\langle G', \Phi(G') \rangle$ as the common subgraph of the $\langle G_0, G_1 \rangle$. 
Let $V' = \{v_1, v_2, ... , v_{|V'|}\}$, the common induced subgraph is also represented as a list of vertex pairs $\{ \langle v_1, \phi(v_1) \rangle, \langle v_2, \phi(v_2) \rangle, ... , \langle v_{|V'|}, \phi(v_{|V'|}) \rangle \}$

The Maximum Common induced Subgraph (MCS) problem aims to find the common induced subgraph such that the number of its vertices is maximized. 
In other words, the two induced subgraphs are isomorphic with the maximum vertex cardinality.
There is a variant problem of MCS, the Maximum Common Connected induced Subgraph (MCCS) problem, which requires the induced subgraph is connected.
%while the latter determines whether $G_0$ (the pattern graph) is isomorphic to an induced subgraph of $G_1$ (the target graph).

%\section{The Branch and Bound for MCS and Baseline Algorithm} 
\section{Branch and Bound for MCS}%Related Work}
This section presents two of the best-performing BnB algorithms for MCS and MCCS,  which are McSplit~\cite{MccreeshPT17} and McSplit+RL~\cite{liu2020learning}. 
To simplify the algorithm description, we initially assume that the graphs are undirected, as the method is easy to be adapted to handle various extensions of the problem~\cite{MccreeshPT17}. 

At each search node, the BnB algorithm first estimates an upper bound of the best solution that can be found in the current subtree. If the upper bound is not larger than the size of the current best solution, the algorithm prunes this branch and backtracks, because any better solution cannot be found %to update the current solution 
under this search node. Otherwise, it selects a new vertex pair to match, updates the current solution and then runs recursively. 
There are three %critical points
key issues for implementing the BnB algorithm: 1) estimate the upper bound, 2) design the branching strategy, 3) maintain candidate vertices of the two graphs. 
To address these issues, McSplit and McSplit+RL proposed label class and a learning based scoring mechanism for the branching nodes, respectively.

%In this section, we introduce the two best-performing BnB algorithms for MCS, McSplit~\cite{MccreeshPT17} and McSplit+RL~\cite{liu2020learning}. McSplit is the base version, which uses an elegant way to solve the above critical points. McSplit+RL proposes a learning based branching heuristic to improve the performance of McSplit, and is the state-of-the-art.

\begin{algorithm}[tb]
\caption{MCS($\mathcal{D}$, $\alpha$, $\beta$, $\gamma$, $S_{cur}$, $S_{best}$)}
\label{alg1}
\textbf{Input}: A bidomain list $\mathcal{D}$; three heuristic functions $\alpha$, $\beta$ and $\gamma$ for selecting the bidomain, the first matched vertex, and the second matched vertex, respectively; the current maintained solution $S_{cur}$; the best solution found so far $S_{best}$.\\
\textbf{Output}: The optimal solution $S_{best}$.

\begin{algorithmic}[1] %[1] enables line numbers
\STATE $UB$ $\leftarrow$ $\text{overEstimate}(\mathcal{D})$ $+$ $|S_{cur}|$
\IF {$UB \leq |S_{best}|$}
\STATE \textbf{return} $S_{best}$
\ENDIF
\STATE $\langle V_l, V_r \rangle$ $\leftarrow$ a bidomain by the heuristic $\alpha(\mathcal{D})$
\STATE $p$ $\leftarrow$ a vertex obtained by the heuristic $\beta(V_l)$
\STATE $m$ $\leftarrow$ $|V_r|$
\FOR {i from 1 to $m$}
\STATE $q$ $\leftarrow$ a vertex obtained by the heuristic $\gamma(V_r)$
\STATE $V_r$ $\leftarrow$ $V_r \setminus \{q\}$
\STATE $S_{cur}'$ $\leftarrow$ $S_{cur} \cup \{ \langle p, q \rangle \}$
\IF {$|S_{cur}'| > |S_{best}|$}
\STATE $S_{best}$ $\leftarrow$ $S_{cur}'$
\ENDIF
\STATE $\mathcal{D}'$ $\leftarrow$ a new bidomain list obtained by dividing $\mathcal{D}$
\STATE $S_{best}$ $\leftarrow$ MCS($\mathcal{D}'$, $\alpha$, $\beta$, $\gamma$, $S_{cur}'$, $S_{best}$)
\ENDFOR
\STATE $\mathcal{D}'$ $\leftarrow$ a new bidomain list obtained by removing $p$ from $\mathcal{D}$ 
\STATE $S_{best}$ $\leftarrow$ MCS($\mathcal{D}'$, $\alpha$, $\beta$, $\gamma$, $S_{cur}$, $S_{best}$)
\STATE \textbf{return} $S_{best}$
\end{algorithmic}
\end{algorithm}

\begin{figure}[t]
\centering
\includegraphics[width=0.9\columnwidth]{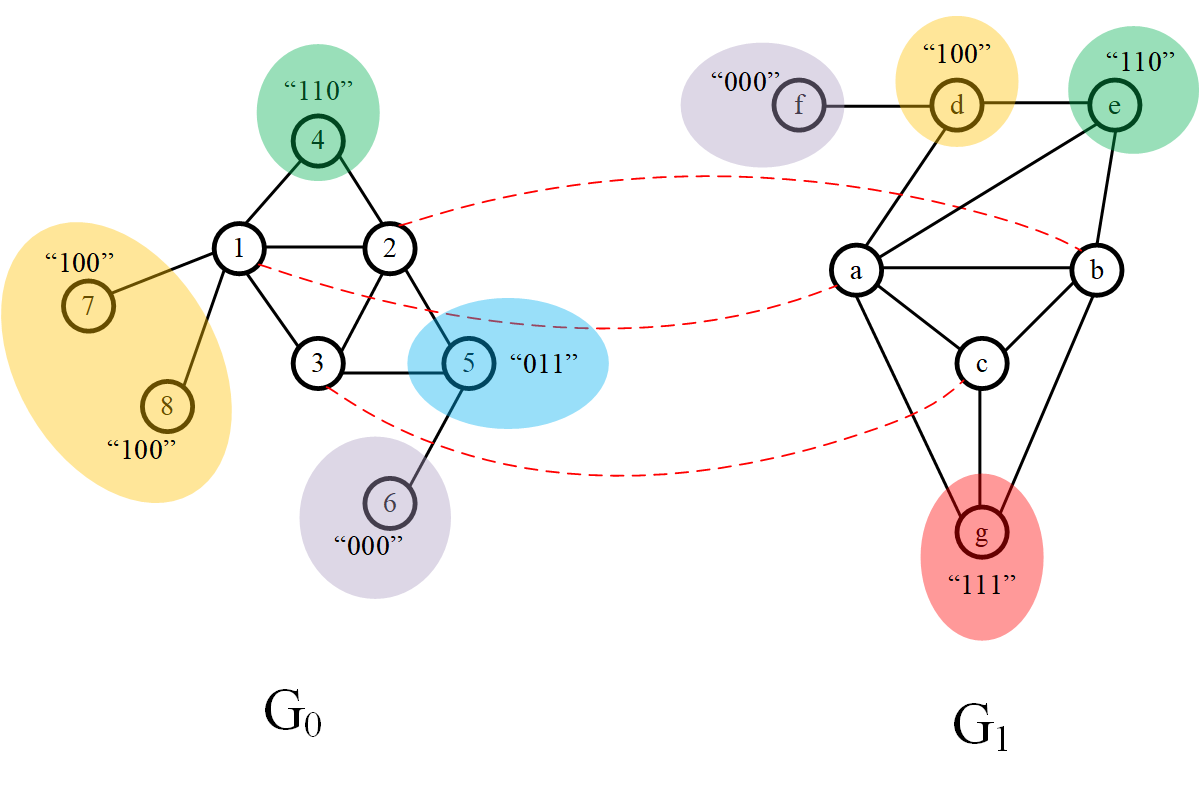}
\vspace{-0.5em}
\caption{ 
An example to illustrate the concepts of bidomain and label class for MCS. There are three matched vertex pairs $\langle 1, a \rangle$, $\langle 2, b \rangle$ and $\langle 3, c \rangle$. According to the definitions of bidomain and label class, the remaining vertices are divided into three valid bidomains $D$(``000'') $= \langle \{6\}, \{f\} \rangle$, $D$(``100'') $= \langle \{7, 8\}, \{d\} \rangle$ and $D$(``110'') $= \langle \{4\}, \{e\} \rangle$, where ``0'' indicates that the vertex is not adjacent to the corresponding matched vertex and ``1'' otherwise.
%indicates the vertex is adjacent to the corresponding matched vertex. 
Note that $D$(``011'') and $D$(``111'') are invalid bidomains, as they cannot provide any matched vertex pair, and will be removed in the bidomain dividing process. 
}
\vspace{-0.5em}
\label{fig1}
\end{figure}

%Before discuss the algorithm, we first introduce two concepts of McSplit, called bidomain structure and label class.
%\subsection{McSplit} \label{Sec3.1}
%\subsection{The General BnB Framework}
\subsection{The BnB Framework for MCS Problems}
\label{Sec3.1}
%The general BnB framework for MCS problems is proposed in \cite{MccreeshPT17}. 
\cite{MccreeshPT17} proposes a BnB framework for MCS problems, that the current best-performing algorithms McSplit and McSplit+RL both adopt.
For two input graphs $G_0$ and $G_1$, the bidomain structure consists of two vertex sets $\langle V_l, V_r \rangle$, where $V_l$ and $V_r$ are subsets of $V(G_0)$ and $V(G_1)$, respectively. McSplit assigns a label to each of the bidomains. 
%the vertices in the two graphs that%
The label indicates that each vertex in the bidomain has the same connectivity to the matched vertices, and is represented by a ``01''-string. %The vertices of two graphs that have the same class label are partitioned into a bidomain,
Therefore the two graphs can be represented as a list of bidomains $\mathcal{D}$ during the search and any vertex pair selected from each bidomain is legal to match.
%different vertex sets 
%in each bidomain.
Whenever a vertex pair $\langle p, q \rangle$ matches, each bidomain $\langle V_l, V_r \rangle$ will be divided into two new bidomains $\langle V_l^0, V_r^0 \rangle$ and $\langle V_l^1, V_r^1 \rangle$, where the superscript ``0'' and ``1'' indicate the connectivity of the matched vertex pair $\langle p, q \rangle$. 
%The following example shows how the bidomain structure and label class work in McSplit. 
See Figure~\ref{fig1} for an illustrative example. 

\iffalse
\begin{example}
Figure~\ref{fig1} shows two undirected and unlabelling graphs $G_0$ and $G_1$. Initially, all vertices of the two graphs form a bidomain with an empty string label $D$(``'') $= \langle \{1, 2, 3, 4, 5, 6, 7, 8\}, \{a, b, c, d, e, f, g\} \rangle$. The first matched vertex pair is $\langle 1, a \rangle$, and the bidomain $D$(``'') divides into $D$(``0'') $= \langle \{5, 6\}, \{f\} \rangle$ and $D$(``1'') $= \langle \{2, 3, 4, 7, 8\}, \{b, c, d, e, g\} \rangle$ where the vertices in the bidomain labelled by ``0'' are not adjacent to the vertices of matched pair and otherwise labelled by ``1''. 
%the vertices in the bidomain labelled by ``1'' are adjacent to one of the vertices of matched pair. 
Then the second matched vertex pair is $\langle 2, b \rangle$. Consequently, the bidomain $D$(``0'') divides into $D$(``00'') $= \langle \{6\}, \{f\} \rangle$ and $D$(``01'') $= \langle \{5\}, \emptyset \rangle$; the bidomain $D$(``1'') divides into $D$(``10'') $= \langle \{7, 8\}, \{d\}\rangle$ and $D$(``11'') $= \langle \{3, 4\}, \{c, e, g\} \rangle$. Note that the bidomain $D$(``01'') has an empty set, indicating that it can not provide any matched vertex pair, thus it is removed from the current bidomain list. After matching the third vertex pair $\langle 3, c \rangle$, the bidomains reaches a %situation
configuration as shown in Figure~\ref{fig1}.
\end{example}
\fi 

Obviously, a bidomain $\langle V_l, V_r \rangle$ can provide at most $\min (|V_l|, |V_r|)$ matched vertex pairs. Therefore, the algorithm estimates the upper bound of the bidomain list $\mathcal{D}$ by the following equation:
\begin{equation} \label{eq1}
    \text{overEstimate}(\mathcal{D}) = \sum_{\langle V_l, V_r \rangle \in \mathcal{D}} \min(|V_l|, |V_r|)
\end{equation}

Then, we introduce the flow of BnB algorithm based on the depth-first search for MCS, as depicted in Algorithm~\ref{alg1}. The call of MCS($\{ \langle V(G_0), V(G_1) \rangle \}$, $\alpha$, $\beta$, $\gamma$, $\emptyset$, $\emptyset$) returns a maximum common induced subgraph of $G_0$ and $G_1$ where $\alpha$, $\beta$ and $\gamma$ are the heuristics. 
At each search node, the algorithm calculates the upper bound assisted by Eq.~\ref{eq1} and uses the upper bound to determine whether it can find a better solution from this search node. If there exists no better solution from this node, the algorithm prunes and backtracks. 
Then it selects a bidomain $\langle V_l, V_r \rangle$ from the current bidomain list $\mathcal{D}$ by heuristic $\alpha$ and picks a vertex $p$ from $V_l$ by heuristic $\gamma$. The algorithm enumerates all the vertices in $V_r$ to match with the vertex $p$ where the order of enumeration is decided by heuristic $\gamma$. When a vertex pair $\langle p, q \rangle$ is matched, the algorithm appends the matched pair to the current solution, updates the best solution, and obtains a new bidomain list $\mathcal{D}'$ by dividing the current bidomain list $\mathcal{D}$ by $\langle p, q \rangle$. Afterwards, the algorithm runs recursively with the new bidomain list $\mathcal{D}'$. After enumerating all the possible matched vertex pairs of $p$, the rest of the configuration is that vertex $p$ does not appear in the matched vertex pair list, thus the algorithm removes vertex $p$ from the current bidomain list and runs recursively. 

\subsection{McSplit and McSplit+RL}
\label{Sec3.2}
There are three heuristics $\alpha$, $\beta$ and $\gamma$ that represent the strategy of selecting the bidomain, selecting the first matched vertex and the order of enumerating the second matched vertex, respectively. 
These heuristics are the essential components in the BnB framework, that lead to different BnB algorithms. %McSplit and McSplit+RL 
%and algorithm that decides the branching strategy of the search tree. The quality of the heuristics will directly affect the search efficiency and algorithm performance. 

McSplit~\cite{MccreeshPT17} implements the three heuristics as follows:
\begin{itemize}
\item $\alpha$: McSplit defines the value of $\max(|V_l|, |V_r|)$ as the size of the bidomain. It selects a bidomain with the smallest size from the bidomain list and uses the largest vertex degree in $V_l$ to break ties. 
%By applying this strategy, the algorithm will branch less at each search node and partitions the bidomains frequently. It certainly reduces the size of the search tree. Note that it can also use value $|V_l| \cdot |V_r|$ as the size of bidomain which means there are $|V_l| \cdot |V_r|$ ways to select a vertex pair from a bidomain, but empirical results show that using $\max(|V_l|, |V_r|)$ is better. 
\item $\beta$: Picks the largest degree vertex from $V_l$ as the first matched vertex. 
\item $\gamma$: Enumerates the second matched vertex in $V_r$ in decreasing order of vertex degree. %nonascending order.% from large to small.  
\end{itemize}

%In McSplit, the selecting vertex pair heuristics $\beta$ and $\gamma$ are naive and not adaptive. 
%\subsection{McSplit+RL} 
%In McSplit, the selecting vertex pair heuristics $\beta$ and $\gamma$ are not adaptive.
McSplit+RL~\cite{liu2020learning} introduces a learning based branching heuristic, 
%to reduce the size of the BnB search tree inspired by reinforcement learning. 
so as to choose a branching pair with the largest bound reduction and reach the pruning condition as early as possible. 
McSplit+RL regards the BnB algorithm as an agent and a vertex pair selection as an action. When a vertex pair $\langle p, q \rangle$ is matched, and the current bidomain list $\mathcal{D}$ is divided into $\mathcal{D}'$,
%the current bidomain list $D$ divides into a new bidomian list $D'$, 
McSplit+RL uses the estimated bound reduction of the bidomain list as the reward $r(p, q)$ of taking action $\langle p, q \rangle$. 
\begin{equation} \label{eq2}
    r(p, q) = \text{overEstimate}(\mathcal{D}) - \text{overEstimate}(\mathcal{D}').
\end{equation}
%where function $\text{overEstimate}(\cdot)$ is defined in Eq.~\ref{eq1}.

McSplit+RL maintains two score lists $S_0(\cdot)$ and $S_1(\cdot)$ for each vertex in $G_0$ and $G_1$, records the accumulated rewards of each vertex. Specifically, update the score lists as follows:
%both are initialized to 0. At each search node, when the algorithm takes an action $\langle p, q \rangle$, McSplit+RL computes the reward $r(p, q)$ immediately and updates the score lists as follows:
\begin{equation}
\begin{aligned}
    S_0(p) &\gets S_0(p) + r(p, q) \\
    S_1(q) &\gets S_1(q) + r(p, q)
\end{aligned}
\end{equation}

Then McSplit+RL replaces heuristics $\beta$ and $\gamma$ based on the two score lists as follows:
%evaluates the branching pair selection that would lead to the largest bound reduction.
%It replaces heuristics $\beta$ and $\gamma$  as follows:
\begin{itemize}
\item $\beta$: Picks a vertex $p$ with the largest score $S_0(p)$ from $V_l$ as the first matched vertex. 
\item $\gamma$: Enumerates the second matched vertex $q$ in $V_r$ 
%with the order is the score $S_1(q)$ from large to small.  
%using the decreasing order of $S_1(q)$ score. 
in decreasing order of the $S_1(q)$ score.
\end{itemize}

\section{The Proposed \name Algorithm} 
We first analyze the limits of the existing branching heuristic, and then we propose a new and strengthened branching heuristic called Long-Short Memory (LSM). Furthermore, we propose a Leaf vertex Union Match (LUM) method for the vertex-vertex mapping based algorithm. %and then we extend this method to the general graph matching. 
We implement both LSM and LUM on top of McSplit, and call the resulting new algorithm \name. 

\subsection{Further Discussion on McSplit(+RL)} 
In McSplit, the selecting vertex pair heuristics $\beta$ and $\gamma$ are straightforward and not adaptive, hence its branching strategy may not be the best choice to minimize the search tree size.

In McSplit+RL, the main idea for the branching is to use the bound reduction to evaluate each vertex.% or vertex pair. 
 Then it selects the vertex with the largest evaluation score to match, aiming to reduce the bound as much as possible and quickly reach the pruning condition. The score is accumulated during the entire BnB algorithm, i.e., it is the summation of all the historical evaluation values. 
However, with the number of recursion and backtracking increases, the situation of the current solution and unmatched vertices (the bidomains) in the two graphs have changed drastically. %, but the historical evaluation values take increasing proportion in the score. 
Accumulating the scores evenly causes a bias caused by a large proportion of historical evaluation values when the current configuration differs from the historical configuration.
%to evaluate the current configuration which is large different to the history configuration.
Hence, we need a mechanism that can reasonably eliminate the influence of the historical evaluations.

%\HK{Besides, McSplit+RL has tried a simple variant of further calculating the reward $r(p, q)$ of a match pair $\langle p, q \rangle$ but its performance is weaker than McSplit+RL. However, we feel the evaluation on the vertex pair is useful because ... The reason that it does not work well in McSplit+RL is due to the ... By ..., we ... }.

\subsection{Long-Short Memory Branching Heuristic} \label{Sec4.2}
To overcome the limits of existing BnB methods,
we apply the reinforcement learning method in McSplit+RL to evaluate the benefit of matching each vertex on the reduction of the search tree size, and further propose a mechanism to eliminate the historical evaluations by decaying part of the evaluation values when they reach a predetermined upper limit.

Our method maintains a score list $S_0(\cdot)$ as McSplit+RL does, and further maintains a score table $S_{t}(p, q)$ of each vertex pair $\langle p, q \rangle$ simultaneously. The initial value of the score table $S_{t}(p, q)$ is set to 0 for each vertex pair $\langle p, q \rangle$. Whenever action $\langle p, q \rangle$ is performed, $S_{t}(p, q)$ will be updated by the following formula:

\begin{equation} \label{eq4}
    S_{t}(p, q) = S_{t}(p, q) + r(p, q).
\end{equation}

%adding the reward $r(p, q)$ (Eq. \ref{eq2}). The $\beta$ heuristic in our method is the same as in McSplit+RL. While we change the $\gamma$ heuristic to: 

Our method replaces the $\gamma$ heuristic in McSplit+RL with:

\begin{itemize}
    \item $\gamma$: Enumerates the second matched vertex $q$ in $V_r$ in decreasing order of the $S_{t}(p, q)$ score.
\end{itemize}

Before introducing our %proposed 
decaying mechanism, we first provide insights on the scoring mechanism. 
Firstly, score $S_0(p)$ is accumulated by reward $r(p, q)$ where vertex $p$ and $q$ are in the same bidomain. The bound reduction is seriously influenced by matching to which vertex $q$, and the bidomains are changed frequently that means the vertices matched to $p$ are very dynamic. It will lead the evaluation score to be outdated quickly. %, the same as to $S_1(q)$. 
Thus, using the score $S_0(p)$ to evaluate the bound reduction of taking action $\langle p, q \rangle$ is very inaccurate.
Secondly, the evaluation score on vertex pair $S_{t}(p, q)$ records the reward $r(p, q)$ properly, which is more accurate than $S_0(p)$. So the historical rewards accumulated by $S_{t}(p, q)$ have significant reference value. So we propose a LSM strategy to make the score list $S_0(\cdot)$ focus on the short-term reward, and the score table $S_{t}(\cdot, \cdot)$ focus on the long-term reward.

%need a mechanism of branching that makes the score list $S_0(\cdot)$ and $S_1(\cdot)$ focus on the short-term reward and the pair score list $S_{t}(\cdot, \cdot)$ focus on the long-term reward.
%The score $S_0(p)$ is accumulated by $R(p, q)$ where vertex $p$ and $q$ are in a same bidomain, and algorithm uses 

%In order to handle the different properties of $S_0(\cdot)$ ($S_1(\cdot)$) and $S_{t}(\cdot, \cdot)$, we use 
In LSM, a short-term threshold value $T_{s}$ ($10^5$ by default) and a long-term threshold value $T_{l}$ ($10^9$ by default) are used to implement our mechanism. At each search node, when the algorithm is branching and the score $S_0(p)$ and $S_{t}(p, q)$ are updated, if score $S_0(p)$ is greater than $T_{s}$, then all the scores in the score list $S_0(\cdot)$ decay to a half.
%Note that the decay operations in $S_0(\cdot)$ and $S_1(\cdot)$ are independent.
And we regard the score table $S_{t}(\cdot, \cdot)$ as $|V(G_0)|$ score lists, i.e., $S_{t}(\cdot, \cdot) = S_{t}(v_1, \cdot), S_{t}(v_2, \cdot), ..., S_{t}(v_{|V(G_0)|}, \cdot)$. If score $S_{t}(p, q)$ is greater than $T_{l}$, then all the scores in the score list $S_{t}(p, \cdot)$ decay to a half. The decaying operation in each score list of the score table is independent.  %By applying such an operation, we can make the score list focus on the short-term reward by setting a small threshold value and focus on the long-term reward by setting a large threshold value. %We call our mechanism the short term memory (on $S_0(\cdot)$ and $S_1(\cdot)$) and long-short term memory (on $S_0(\cdot)$ and $S_{t}(\cdot, \cdot))$) for the branching heuristic. %Note that~\cite{liu2020learning} also proposed a variant of McSplit+RL that applies the 

\subsection{Leaf Vertex Union Match}
We first provide a definition of the leaf vertex of an undirected graph. 
%For convenience, we call the leaf node (vertex) leaf. 
A vertex is regarded as a leaf if it is adjacent to exactly one vertex in the graph and the leaves of vertex $u$ indicate the leaves adjacent to $u$. 
Then, we provide the main theory to support our LUM strategy. 

\begin{theorem}
In the BnB framework, when a vertex pair $\langle p, q \rangle$ is matched, we can match as many leaf pairs as possible from the unmatched leaves of vertex $p$ and $q$ without affecting the solution's optimality. 
\end{theorem}

\begin{proof}
Assume the current solution has $k$ matched vertex pairs and the $i$-th pair is $\langle p, q \rangle$. According to the label class definition in Section~\ref{Sec3.1}, the labels of the unmatched leaves of $p$ and $q$ are all the same as ``0...010...0'' which has exactly one ``1'' in the $i$-th position, because the leaves are only adjacent to the vertex $p$ and $q$ respectively. Thus, these leaves are always partitioned in the same bidomain. It guarantees that no matter what the matching configuration of other vertices is, the leaves of $p$ and leaves of $q$ are always legal to match as long as $p$ and $q$ match together. If arbitrary pair of leaves are matched, all the remaining unmatched vertices are not adjacent to both matched leaves. Matching any leaf pair does not divide a new bidomain from the current bidomain list. Therefore, such an operation does not affect the solution's optimality of the BnB algorithm for MCS. 
\end{proof}

\begin{figure}[t]
\centering
\includegraphics[width=0.9\columnwidth]{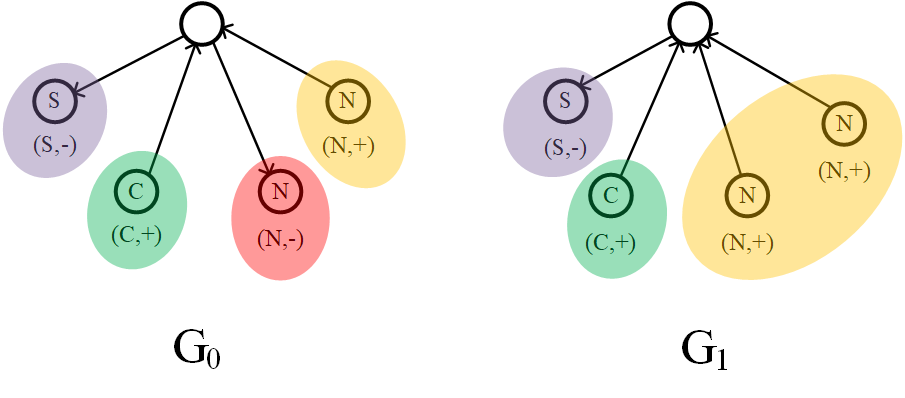}
\vspace{-0.5em}
\caption{ 
An example to illustrate the leaf attribute of LUM for a general situation. It shows two directed and labelling graphs $G_0$ and $G_1$ (only vertices are labelled, thus we omit the edge label in the leaf attribute), each leaf attribute is given below the leaf. The leaf attribute $(A, \pm)$ represents the vertex label and the direction of arc connected to the leaf respectively. Sign $+$ means the leaf is the head of the arc and sign $-$ means the leaf is the tail of the arc. There are at most three leaf pairs that can be provided in this example. 
}
\vspace{-0.5em}
\label{fig2}
\end{figure}

\vspace{-1em}
Our LUM strategy can be extended to general graph matching problems. Consider a directed and labelling graph, where both vertices and arcs are labelled. Let $L_v = \{a_1, a_2, a_3, ...\}$ and $L_e = \{b_1, b_2, b_3, ...\}$ represent the vertex label set and the arc label set, respectively. We give each leaf an attribute $(a_i, \pm b_j)$ where $a_i$ is the label of vertex, $b_j$ is the label of arc connected to the leaf and the sign represents the direction of arc ($+$ means leaf is head, $-$ means leaf is tail). We partition the leaves into different groups according to the attributes. Leaves in the corresponding group can be matched, therefore it can match as many leaf pairs as possible in each group. We also provide an example as shown in Figure~\ref{fig2}. Note that the leaf grouping is independent of the bidomain dividing.

\section{Experiments}%al Results}
We compare our proposed \name algorithm with McSplit+RL~\cite{liu2020learning}, which is the state-of-the-art BnB exact algorithm for MCS. We tested the algorithms on both the MCS and MCCS problems. The experimental results show that \name outperforms McSplit+RL for solving these problems. Further ablation studies show the effectiveness of our proposed methods, including LSM branching heuristic and LUM strategy.

\subsection{Experimental Setup}
Experiments were performed on a server with Intel® Xeon® E5-2650 v3 CPU and 256GBytes RAM. The algorithms were implemented in C++ and compiled using g++ 5.4.0. The cutoff time is set to 1800 seconds for each instance, which is the same as in~\cite{liu2020learning}. 

There are two parameters, $T_{s}$ and $T_{l}$ for the LSM branching heuristic. We sample some instances solved between 10 to 30 minutes and the parameter tuning domain is set to $10^k, k \in [2, 9]$. Finally, the default setting of two parameters $T_{s}$ and $T_{l}$ are $10^5$ and $10^9$, respectively.

\begin{figure}[t]
\centering
\includegraphics[width=0.98\columnwidth]{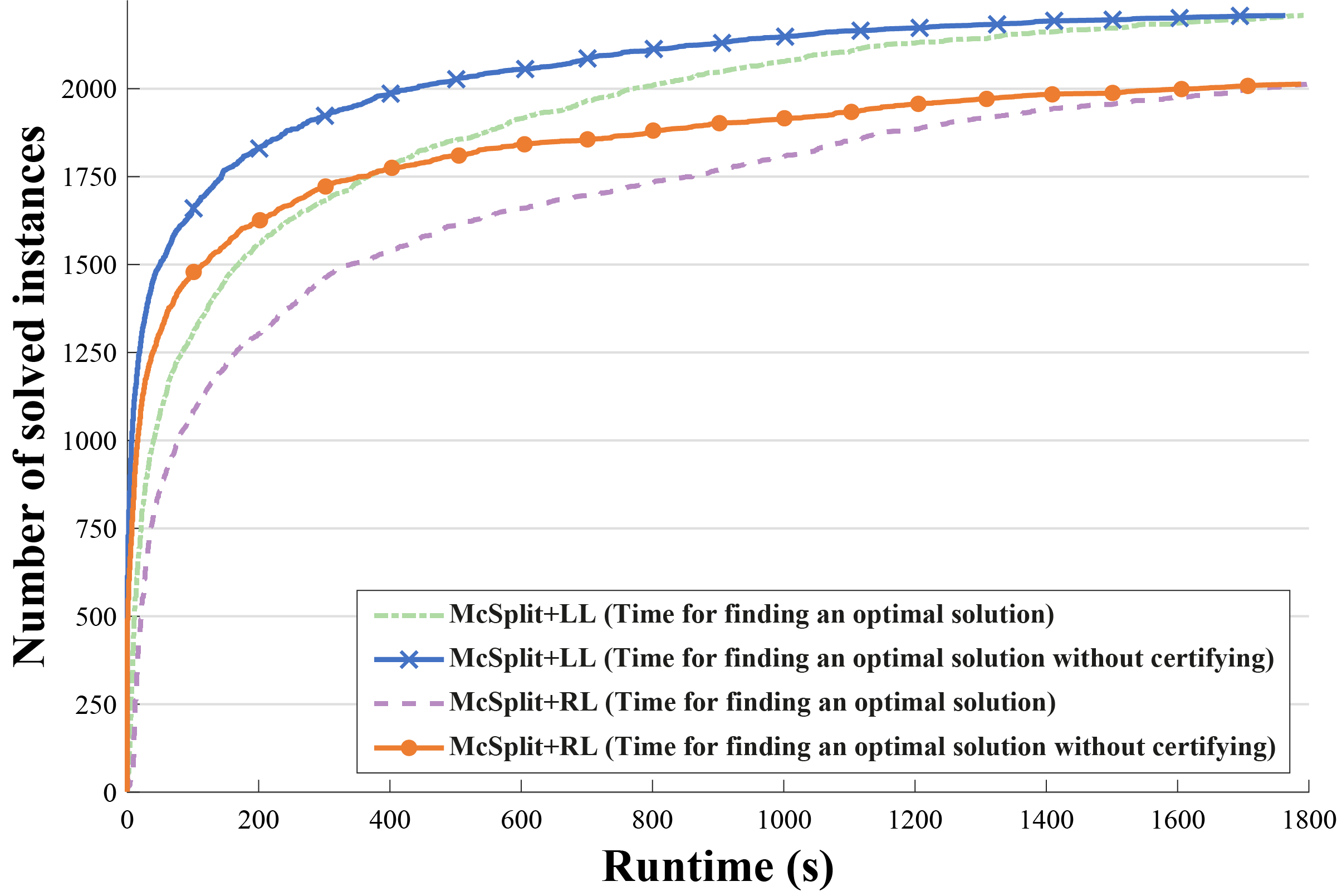}
\caption{ 
Cactus plots of McSplit+RL and McSplit+LL on the 2,309 moderate MCS instances, which solve 2,013 and 2,208 out of the instances respectively. 
}
\vspace{-1em}
\label{cmp1}
\end{figure}

\setcounter{figure}{4}
\begin{figure}[t]
\centering
\includegraphics[width=0.98\columnwidth]{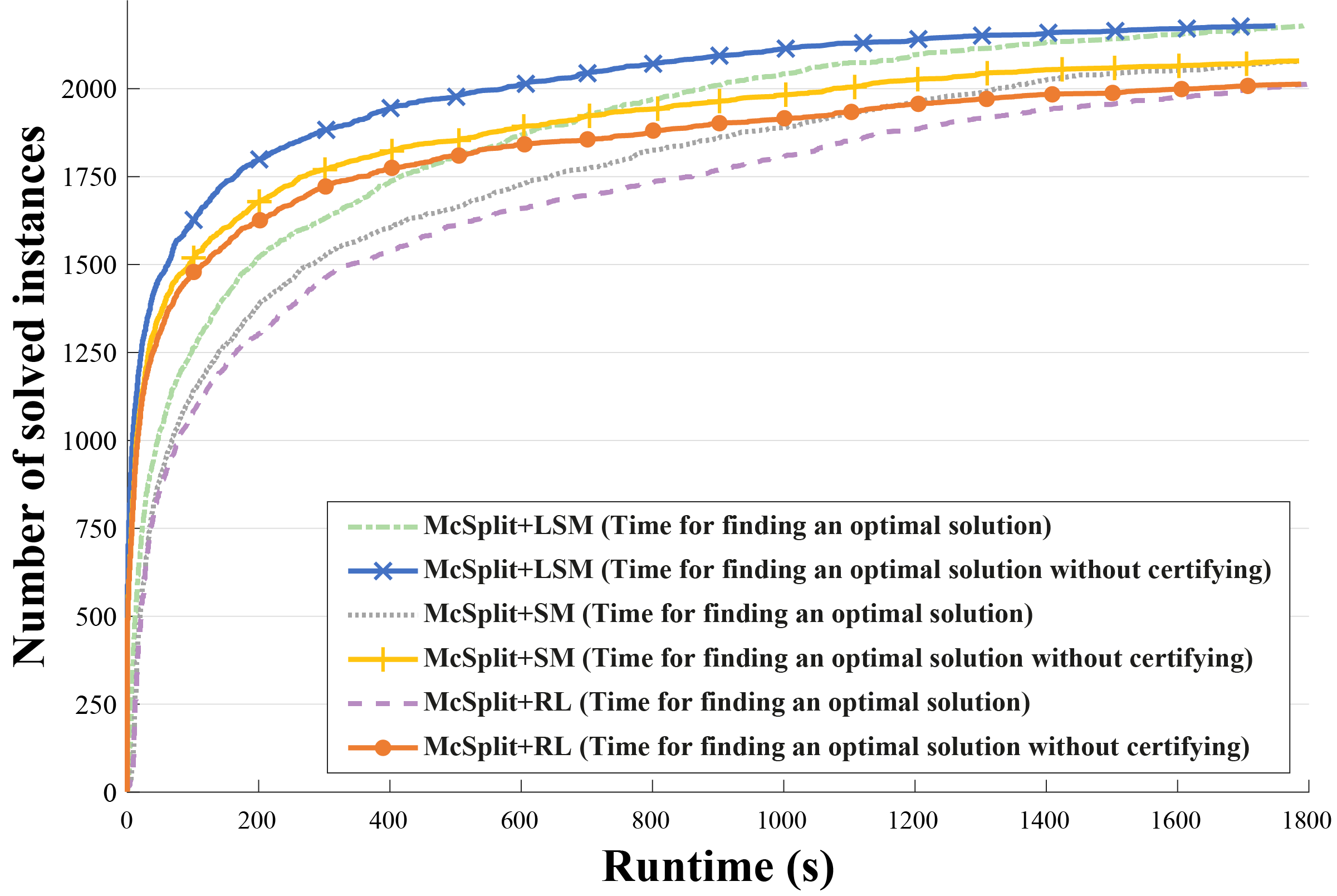}
\caption{ 
Cactus plots of McSplit+RL, McSplit+SM and McSplit+LSM on the 2,309 moderate  MCS instances, which solve 2,013, 2,079 and 2,179 out of the instances respectively.
}
\vspace{-1em}
\label{cmp3}
\end{figure}

\setcounter{figure}{3}
\begin{figure}[t]
\centering
\includegraphics[width=0.98\columnwidth]{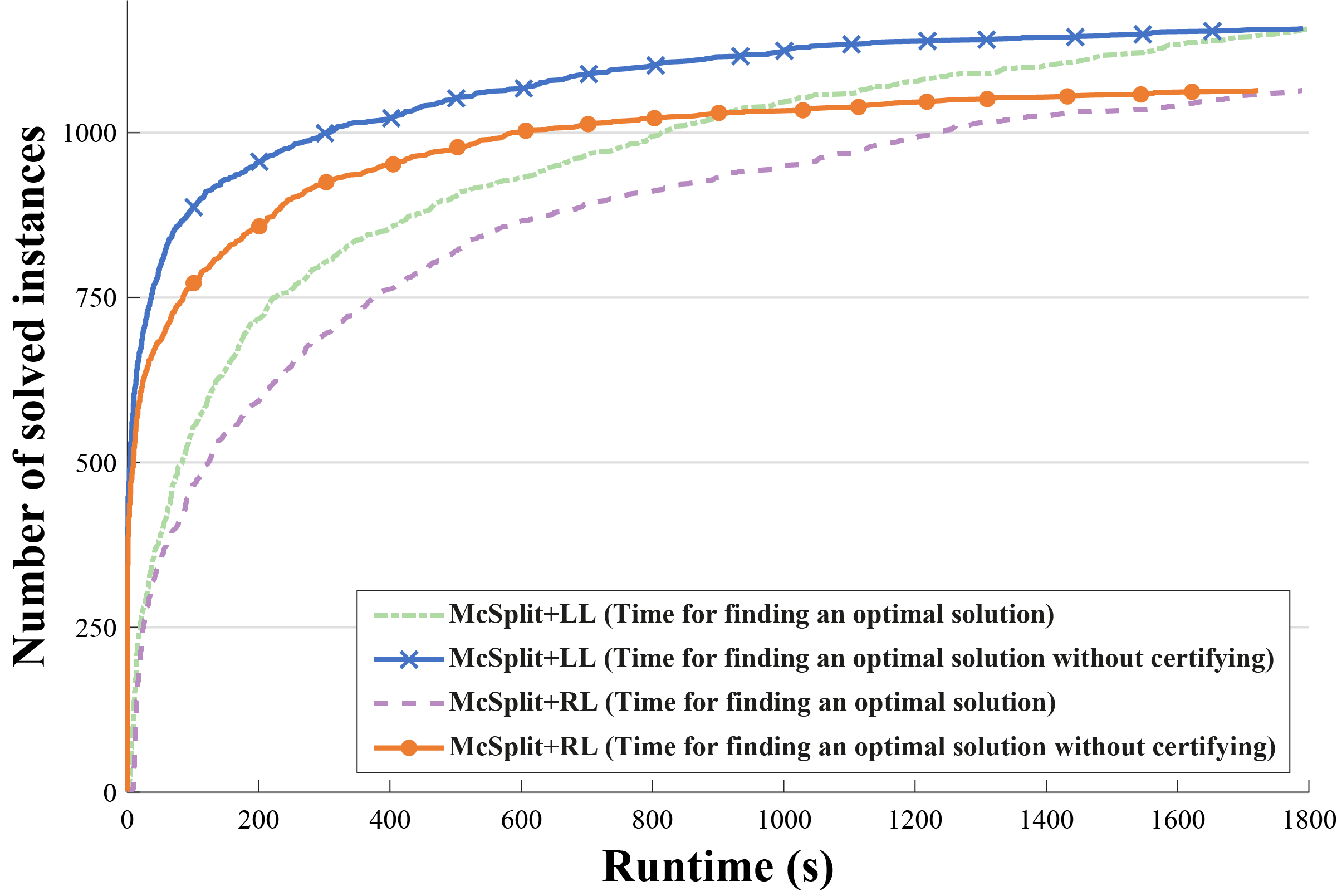}
\caption{ 
Cactus plots of McSplit+RL and McSplit+LL on the 1,166 moderate MCCS instances, which solve 1,064 and 1,158 out of the instances respectively.
}
\vspace{-1em}
\label{cmp2}
\end{figure}

\setcounter{figure}{5}
\begin{figure}[t]
\centering
\includegraphics[width=0.98\columnwidth]{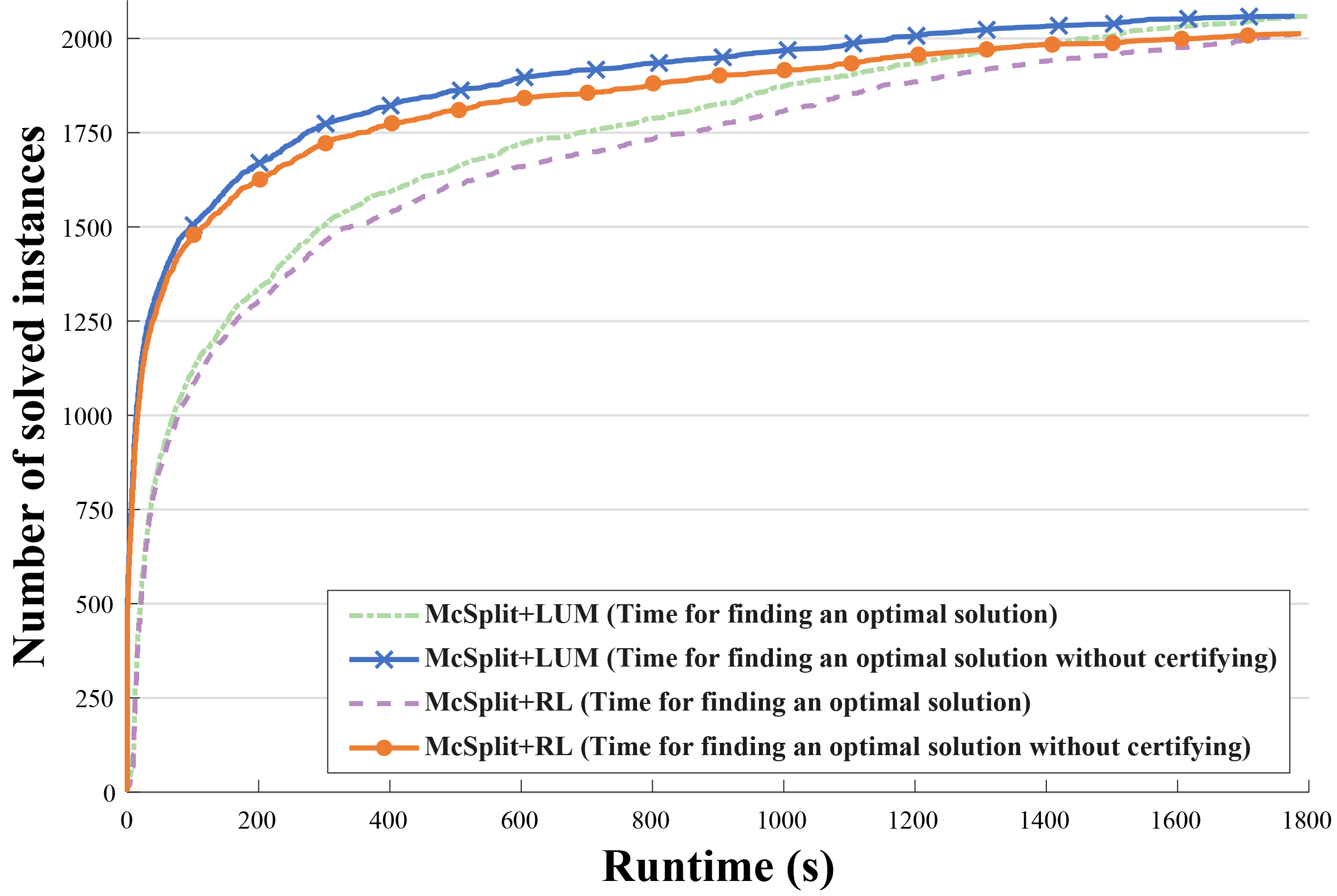}
\caption{ 
Cactus plots of McSplit+RL and McSplit+LUM on the 2,309 moderate MCS instances, which solve 2,013 and 2,059 out of the instances respectively.
}
\vspace{-1em}
\label{cmp4}
\end{figure}

\subsection{Benchmark Datasets}
The benchmark datasets\footnote{Available at http://liris.cnrs.fr/csolnon/SIP.html} include 25,552 instances that are divided into several compositions: % as follows:
\begin{itemize}
\item Biochemical reaction~\cite{gay2014subgraph} includes 136 directed bipartite graphs (with vertices between 9 and 386), and describe biochemical reaction networks originated from the biomodels.net. 
It provides 9,180 instances obtained by pairing each of the graphs. 
\item Images-PR15~\cite{solnon2015complexity} includes a target graph (with 4,838 vertices) and 24 pattern graphs (with vertices between 4 and 170), which are generated from segmented images. It provides 24 instances. 
\item Images-CVIU11~\cite{damiand2011polynomial} includes 43 pattern graphs (with vertices between 22 and 151), and 146 target graphs (with vertices between 1,072 and 5,972), which are generated from segmented images. It provides 6,278 instances. 
\item Meshes-CVIU11~\cite{damiand2011polynomial} includes 6 pattern graphs (with vertices between 40 and 199), and 503 target graphs (with vertices between 208 and 5,873), which are generated from meshes modelling 3D objects. It provides 3,018 instances. 
\item Scalefree~\cite{zampelli2010solving,solnon2010alldifferent} includes 100 instances. Each instance is composed of a target graph (with vertices between 200 and 1,000) and a pattern graph (contains 90\% of the vertices of the target graph). These instances are randomly generated from scale-free networks.
\item Si~\cite{zampelli2010solving,solnon2010alldifferent} includes 1,170 instances. Each instance is composed of a target graph (with vertices between 200 and 1,296) and a pattern graph (with vertices between 20\% and 60\% of the target graph). Among these instances, there are bounded valence graphs and modified bounded valence graphs, 4D meshes, and randomly generated graphs.
\item LV~\cite{MccreeshPT17} includes 49 pattern graphs and 48 target graphs (with vertices both between 10 and 128). It provides 2,352 instances.  
\item LargerLV~\cite{MccreeshPT17} includes 49 pattern graphs (with vertices between 10 and 128) and 70 target graphs (with vertices between 138 and 6,671). It provides 3,430 instances.
\end{itemize}

\subsection{Comparisons on MCS}
The McSplit algorithm is efficient enough and there are lots of small scale instances in the benchmark datasets that can be solved in several seconds. 
The results on these small scale instances cannot really show the gap of differences on the algorithm efficiency.  
%provide the quality reference on the algorithm efficiency.
%, but we also give the average time of solving the small scale instances. 
Hence, we group the 7,226 small scale instances which could be solved by all the tested algorithms within 10 seconds into the easy set, and only provide the average solving time.  
We also exclude another set of tough instances that cannot be solved by any algorithm within the time limit. 
Thereafter, we have 2,309 remaining instances that can be solved by at least one of the tested MCS algorithms within the time limit. 
We denote these medium hard instances as the moderate instances, and will use them for detailed performance comparison.

We compare McSplit+RL with \name
%and the results are 
as illustrated in Figure~\ref{cmp1}. 
Each point ($t$, $n$) in a curve of Figure~\ref{cmp1} indicates the algorithm solves (finds the optimal solution, with or without certifying) $n$ instances in $t$ seconds, the same in Figures~\ref{cmp2}, Figure~\ref{cmp3}, and Figure~\ref{cmp4}.
%and the same as meaning in the rest of comparison figures.
McSplit+RL solves 2,013 moderate instances while \name could solve 2,208 moderate instances. 
Besides, the average runtimes of McSplit+RL and \name on %the small scale instances 
the easy set of instances are 0.83s and 0.51s, respectively. 
In other words, \name solves 9.69\% more moderate instances than McSplit+RL, and \name is also faster than McSplit+RL in solving the easy instances. 
The results demonstrate that \name outperforms McSplit+RL significantly on the MCS problem.

\subsection{Comparisons on MCCS}
We also apply our \name algorithm to solve the MCCS problem. As the basic version of McSplit does not support solving MCCS on the directed graph, we exclude the directed graph instances (9,180 Biochemical reaction instances) from the datasets. 

%As same as the comparison on MCS,  
Using the same datasets processing as in MCS, %we first exclude 2,110 easy and tough instances
%\HK{we first exclude the ??? easy instances and the ??? tough instances},
we first exclude the 2,110 easy instances and the tough instances, 
and use the remaining 1,166 moderate instances as benchmarks for detailed comparison. 
%then we compare McSplit+LL with McSplit+RL on the remainder 1,166 instances and 
The results are illustrated in Figure~\ref{cmp2}. McSplit+RL solves 1,064 moderate instances, while \name solves 1,158 moderate instances (8.83\% more than McSplit+RL).
Also, the average runtimes of McSplit+RL and \name on %the small scale instances 
 the easy instances are 2.03s and 1.36s, respectively. 
The results demonstrate that McSplit+LL also outperforms McSplit+RL significantly on the MCCS problem.

\subsection{Ablation Study}
In this subsection, we do ablation studies to analyze the effectiveness of the two proposed operators, LSM and LUM. % on the MCS instances. 

We first compare three algorithms, McSplit+RL, McSplit+SM and McSplit+LSM, on the MCS instances and the results are illustrated in Figure~\ref{cmp3}. McSplit+SM is a variant of McSplit+RL that applies our decaying operation on the score lists $S_0(\cdot)$ and $S_1(\cdot)$ with the short-term threshold value $T_s$, called Short Memory (SM). McSplit+LSM is a variant that applies LSM on the top of McSplit.

McSplit+SM solves 2,079 moderate instances and McSplit+LSM solves 2,179 moderate instances which are 3.28\% and 8.25\% more than McSplit+RL. Besides, the average runtimes of McSplit+SM and McSplit+LSM on the easy instances are 0.82s and 0.53s, respectively. The results show that McSplit+RL can be improved by applying our proposed decaying mechanism,
and using LSM instead of SM yields better performance.

Then, we analyze the effectiveness of LUM. We implement it on top of McSplit+RL and obtain a variant algorithm called McSplit+LUM. We compare McSplit+LUM with McSplit+RL on the MCS instances and the results are illustrated in Figure~\ref{cmp4}. It solves 2,059 moderate instances which is 2.29\% better than McSplit+RL. 
Besides, the average runtime of McSplit+LUM on the easy set of instances was 0.81s. 

The ablation studies for both McSplit+LSM and McSplit+LUM demonstrate that our proposed operators could improve the performance and efficiency of the BnB algorithm for both MCS and MCCS problems. 

%\subsection{Further Analysis}

\section{Conclusion}

In this work, we address the Maximum Common induced Subgraph (MCS) and Maximum Common Connected induced Subgraph (MCCS) problems, and we propose an effective Branch-and-Bound (BnB) algorithm called McSplit+LL for these two problems. 
McSplit+LL incorporates two proposed operators into the effective BnB algorithm McSplit. The first one is a new branching operator on the BnB algorithm called Long-Short Memory (LSM). LSM finds the different properties of scoring each vertex and each vertex pair, and applies a reasonable policy to make the score of each vertex focus on the short-term reward and the score of each vertex pair focus on the long-term reward. The second one is a general operator called Leaf vertex Union Match (LUM) for MCS and MCCS. LUM allows the common subgraph match multiple leaf pairs whenever a vertex pair is matched, so as to speed up the overall matching at the same time not affect the solution optimality. Both LSM and LUM can improve the efficiency of the BnB algorithm. Besides, LUM is a general vertex-vertex mapping method and could be applied to other graph matching problems.

%LSM finds the different properties of each vertex and vertex pair on the branching node evaluation and applies a reasonable policy to make the each vertex score focus on the short-term reward and the vertex pair score focus on the long-term reward. 

%Secondly, we propose a Leaf vertex Union Match (LUM) operator. LUM allows subgraph match multiple leaf pairs for each matched vertex pair, so as to speed up the overall matching at the same time not affect the solution optimality. LUM is a general vertex-vertex mapping method and so it can be applied to other graph matching problems. 
%We incorporate these two operators LSM and LUM into the effective BnB algorithm McSplit. The resulting algorithm is called McSplit+LL. 
We do extensive experiments on %25,552
public instances to evaluate the performance of our proposed algorithm McSplit+LL, and the effectiveness of the two proposed operators LSM and LUM. The results show that McSplit+LL significantly outperforms the best-performing algorithm McSplit+RL for both MCS and MCCS, and both LSM and LUM can improve the efficiency and performance of the BnB algorithm.
%can improve the performance of BnB algorithm for MCS significantly. 

%% The file named.bst is a bibliography style file for BibTeX 0.99c
\clearpage
\bibliographystyle{named}
\bibliography{ijcai22}

\end{document}